\documentclass[12pt,reqno]{article}

\usepackage[usenames]{color}
\usepackage{amssymb}
\usepackage{amsmath}
\usepackage{amsthm}
\usepackage{amsfonts}
\usepackage{amscd}
\usepackage{graphicx}

\usepackage[colorlinks=true,
linkcolor=webgreen,
filecolor=webbrown,
citecolor=webgreen]{hyperref}

\definecolor{webgreen}{rgb}{0,.5,0}
\definecolor{webbrown}{rgb}{.6,0,0}

\usepackage{color}
\usepackage{fullpage}
\usepackage{float}

\usepackage{graphics}
\usepackage{latexsym}
\usepackage{epsf}
\usepackage{breakurl}

\DeclareMathOperator{\sw}{sw}
\DeclareMathOperator{\lw}{lw}

\def\andd{ \, \wedge \, }
\def\orr{ \, \vee \, }
\DeclareMathOperator{\per}{per}
\DeclareMathOperator{\rep}{rep}
\DeclareMathOperator{\firstfac}{firstfac}
\DeclareMathOperator{\twofac}{twofac}

\begin{document}

\theoremstyle{plain}
\newtheorem{theorem}{Theorem}
\newtheorem{corollary}[theorem]{Corollary}
\newtheorem{lemma}[theorem]{Lemma}
\newtheorem{proposition}[theorem]{Proposition}

\newtheorem{definition}[theorem]{Definition}
\theoremstyle{definition}
\newtheorem{example}[theorem]{Example}
\newtheorem{conjecture}[theorem]{Conjecture}

\theoremstyle{remark}
\newtheorem{remark}[theorem]{Remark}

\title{Repetition factorization of automatic sequences}

\author{
Jeffrey Shallit and Xinhao Xu \\
School of Computer Science \\
University of Waterloo \\
Waterloo, ON  N2L 3G1 \\
Canada\\
\href{mailto:shallit@uwaterloo.ca}{\tt shallit@uwaterloo.ca} \\
\href{mailto:x427xu@uwaterloo.ca}{\tt x427xu@uwaterloo.ca}
}

\maketitle

\begin{abstract}
Following Inoue et al., we define a word to be a {\it repetition} if it is a (fractional) power of exponent at least $2$.  A word has a {\it repetition factorization} if it is the product of repetitions.  We study repetition factorizations in several (generalized) automatic sequences, including the infinite Fibonacci word, the Thue-Morse word, paperfolding words, and the Rudin-Shapiro sequence.
\end{abstract}

\section{Introduction}

Let $\Sigma$ be a finite alphabet and let
$w \in \Sigma^+$.  We let $w[i]$ denote the
$i$'th letter of the finite word $w$, where we index
starting at position $1$.   We say that an integer $p \geq 1$ is a {\it period\/} of $w$ if
$p\leq |w|$ and $w[i]=w[i+p]$ for $1 \leq i \leq |w|-p$.   For example, the integers $3,6,$ and $7$ are
all periods of $w = {\tt alfalfa}$.

Following Inoue et al.~\cite{Inoue:2022}, we call a nonempty word $w$ a {\it repetition\/} if $w$ has a period $p \leq |w|/2$.  For example, {\tt alfalfa} is a repetition, but {\tt salsa} is not.

A {\it repetition factorization\/} of a word $w$
is a nonempty list of nonempty words $w_1, w_2, \ldots, w_t$ such
that $w = w_1 w_2 \cdots w_t$ and each $w_i$ is
a repetition.    The {\it width\/} of such a factorization
is the integer $t$, the number of terms appearing
in it.
Some words have no repetition
factorization at all, while others have multiple such
factorizations---of the same or even different widths.  For example,
$aaaababa = (aaaa)(baba)= (aaa)(ababa) = (aa)(aa)(baba)$.    For a given
finite word $w$, two factorization widths are of particular interest:   the shortest and the longest.  We denote them by $\sw(w)$ and $\lw(w)$,
respectively.  If $w$ has no repetition factorization,
then we define $\sw(w) = \lw(w) = 0$.

Recently Inoue et al.~\cite{Inoue:2022} gave an efficient algorithm
to compute repetition factorizations, and they
also studied the repetition factorizations of
Fibonacci words.  These are the finite words defined
by $F_1 = {\tt 1}$, $F_2 = {\tt 0}$, and
$F_n = F_{n-1} F_{n-2}$ for $n \geq 3$.   They also determined the shortest and longest possible widths of the factorizations of Fibonacci words.  This study was followed up by Kishi et al.~\cite{Kishi:2023}.

Inspired by their work,
our goal in this paper is to study repetition
factorizations in more generality.   This generalization is in two different ways.  Instead of just the Fibonacci word, we study automatic sequences more generally.  And instead of just {\it prefixes\/} of specific lengths, we study {\it all factor\/}s of automatic
sequences.

One of the tools we use in this paper is {\tt Walnut}, a program to prove or disprove first-order statements about automatic sequences.  For more about {\tt Walnut}, see \cite{Mousavi:2016,Shallit:2022}.

\section{Results for (generalized) automatic sequences}

Recall that an infinite word (or infinite sequence; for us the terms are synonymous)  ${\bf a} = (a_n)_{n \geq 0}$ is said to be $k$-automatic, for $k \geq 2$ an integer, if there exists a deterministic finite automaton with output (DFAO) that, on input $n$ expressed in base-$k$, ends in a state with output $a_n$.  For more information about automatic sequences, see, for example, \cite{Allouche&Shallit:2003}.  Typical examples of sequences in this class are the Thue-Morse sequence \cite{Thue:1912,Berstel:1995} and the Rudin-Shapiro sequence \cite{Rudin:1959,Shapiro:1952}.

This concept can be generalized further, to sequences where the input $n$ is specified in a more exotic numeration system where the addition relation can be computed by a finite automaton.   A typical example of a sequence in this class is the infinite binary Fibonacci word \cite{Berstel:1980b,Berstel:1986b} where the underlying numeration system is the so-called Zeckendorf system.

The following theorem is the basis for the results in this paper.
\begin{theorem}
\leavevmode
Suppose $\bf x$ is a generalized automatic sequence.
\begin{itemize}
    \item[(a)]
    For every finite $B$ there is a DFAO that, on input $i$ and $n$, computes the width of the  
 shortest repetition
    factorization of ${\bf x}[i..i+n-1]$, provided it is  at most $B$.  Furthermore,
    there is an algorithm to compute this DFAO.

\item[(b)] Furthermore, if there exists a universal bound $B'$ such that $\sw(w) \leq B'$ for all factors (resp., all prefixes) $w$ of $\bf x$, then there is an algorithm to determine this $B'$.

\item[(c)] For every finite $t$ there is a DFAO that, on input $i$ and $n$,  determines if ${\bf x}[i..i+n-1]$ has  a {\rm unique\/} repetition factorization of width exactly $t$ (resp., of width at most $t$).

\end{itemize}
    \label{thm1}
\end{theorem}

The main tool we use to prove this theorem is a classic result due, in its essential points, to B\"uchi \cite{Buchi:1960} and later corrected and extended by
others \cite{Bruyere&Hansel&Michaux&Villemaire:1994}:
\begin{theorem}
Let $\varphi$ be a first-order logical formula involving the logical operations, addition and comparison of natural numbers, and indexing into a (generalized) automatic sequence.   There is an algorithm that produces an automaton $A$ with the following properties:   if $\varphi$ has free variables, then $A$ accepts exactly the values of the free variables that make $\varphi$ true.   If $\varphi$ has no free variables, then the algorithm returns either {\tt TRUE} or {\tt FALSE}.
\end{theorem}

\begin{proof}[Proof of Theorem~\ref{thm1}.]
For the first claim, it suffices to create a first-order logical formula asserting that the factor ${\bf x}[i..i+n-1]$ has a factorization of width $\leq B$.    We do this in several steps.

First, let us create a formula $\per(i,n,p)$ that asserts that $p$ is a period of the factor ${\bf x}[i..i+n-1]$:
$$ p\geq 1 \andd p \leq n \andd  \forall j \  (j+p<n) \implies {\bf x}[i+j]={\bf x}[i+j+p].$$

Next, let's create a formula  $\rep_1(i,n)$ that asserts that 
${\bf x}[i..i+n-1]$ is a repetition:
$$ \exists p \ \per(i,n,p)  \andd 2p \leq n .$$

Next, let's create a series of formulas 
$\rep_2(i,n)$, $\rep_3(i,n)$, and so forth, such that
$\rep_t(i,n)$ asserts that ${\bf x}[i..i+n-1]$ is a concatenation of $t$ or fewer repetitions, for $t \geq 2$:
$$ \rep_{t-1}(i,n) \orr \exists j\  j\geq 1 \andd j<n \andd \rep_1(i,j)  \andd
\rep_{t-1}(i+j,n-j) .$$

For the first claim, it suffices to create the corresponding automata up to $t= B$.  Then the automaton of least index accepting $(i,n)$ corresponds to 
the width of ${\bf x}[i..i+n-1]$.    Using a standard product construction, we can then combine the $B'$ different automata into one DFAO that  on input $(i,n)$ computes the least $t$ for which $\rep_t(i,n)$ holds.

For the second claim, we create the series of automata $\rep_t(i,n)$ for $t = 1,2,3,\ldots$, checking at every step if $\rep_t(i,n)$ and $\rep_{t+1}(i,n)$ recognize the same language.
We can do this with the following formula:
$$ \forall i, n\  \rep_t(i,n) \iff \rep_{t+1}(i,n).$$
If they do recognize the same language, then we can take $B' = t$.  
To see that this works, 
suppose, contrary to what we want to prove, that for at least one factor $y := {\bf x}[i..i+n-1]$  with a repetition factorization, at least width $t+1$ is required. Then $\rep_{t+1}(i,n)$ accepts $y$ and hence so does $\rep_t(i,n)$, a contradiction.  

Exactly the same ideas work for prefixes alone, instead of arbitrary factors.

\bigskip
(c)  Suppose there are at least two different width-$t$ different factorizations of ${\bf x}[i..i+n-1]$.   Then there would be two sets of indices
$i_1 = i < i_2 < \cdots < i_{t+1} =i+n$
and
$j_1 = i < j_2 < \cdots < j_{t+1} =i+n$ 
such that 
${\bf x}[i_s..i_{s+1}-1]$
and ${\bf x}[j_s..j_{s+1}-1]$
are all repetitions for 
$1 \leq s \leq t$ and
$i_s \not= j_s$ for at least
one $s$.   It is easy to write a first-order formula expressing this.   To handle the case of factorizations of width $\leq t$, we additionally perform the same check for every $t'<t$, and then also check the existence of factorizations of different two different widths.
\end{proof}

\begin{remark}
If there is a universal bound $B$ on the width of factors of $\bf x$, then we can check  the uniqueness of the repetition factorization of {\it all\/} factors as follows:  ensure there are no factors with two repetition factorizations, where the first term is of different lengths.   This is clearly necessary.  To see it is sufficient, assume $y$ is a shortest factor with two different repetition factorizations.  Then if the first term is of the same length in both, we could delete it to get a shorter counterexample.  

First let us write a formula
$\firstfac(i,n,t)$ asserting
that ${\bf x}[i..i+n-1]$ has a repetition factorization where the first term is of length $t$:
$$ 1\leq t \andd t\leq n \andd \rep_1(i,t) \andd
\rep_{B-1}(i+t,n-t).$$
Next, we assert that there are two such factorizations with different $t$:  
$$\twofac(i,n) := \exists t_1, t_2 \ (t_1 \not= t_2) \andd \firstfac(i, n, t_1) \andd \firstfac(i,n, t_2) .$$
Finally, the formula
$$ \neg\exists i, n \ \twofac(i,n)$$
guarantees uniqueness.
\label{rmk1}
\end{remark}

In the remaining sections, we apply these ideas to various infinite words of interest.

\section{Results about factors of the Fibonacci word}

The famous infinite Fibonacci word
${\bf f} = 01001010\cdots$ can be defined in many
different ways.  Possibly the easiest is as the
fixed point of the morphism $0 \rightarrow 01$,
$1 \rightarrow 0$.

\begin{theorem}
\leavevmode
\begin{itemize}
\item[(a)] 
Every finite factor of {\bf f}, if it has a repetition factorization, has
one of width at most $3$.
\item[(b)] Every finite prefix of {\bf f}, if it has a repetition factorization, has
one of width at most $2$.
\end{itemize}
\end{theorem}

\begin{proof}
(a):  
It suffices to translate the logical formulas of the previous section into {\tt Walnut} code, and then check the results.  Here `{\tt rep}$n${\tt \_f}' asserts that ${\bf f}[i..i+n-1]$ is
the concatenation of at most $n$ repetitions.  

To understand the code that follows, note that {\tt \&} is logical AND, {\tt A} and {\tt E} represent the universal and existential quantifiers, respectively, {\tt =>} is logical implication, 
{\tt <=>} is logical IFF, {\tt |} is logical OR, and {\tt F} is {\tt Walnut}'s way to represent the infinite Fibonacci word, indexed starting at position $0$.
\begin{verbatim}
def per_f "?msd_fib p>=1 & p<=n & Aj (j+p<n) => F[i+j]=F[i+j+p]":
# 42 states

def rep1_f "?msd_fib Ep $per_f(i,n,p) & 2*p<=n":
#16 states

def rep2_f "?msd_fib $rep1_f(i,n) | Ej j>=1 & j<n & $rep1_f(i,j) & 
   $rep1_f(i+j,n-j)":
#64 states

def rep3_f "?msd_fib $rep2_f(i,n) | Ej j>=1 & j<n & $rep1_f(i,j) & 
   $rep2_f(i+j,n-j)":
#35 states
   
def rep4_f "?msd_fib $rep3_f(i,n) | Ej j>=1 & j<n & $rep1_f(i,j) & $rep3_f(i+j,n-j)":
#35 states

eval thm3a "?msd_fib Ai,n $rep3_f(i,n) <=> $rep4_f(i,n)":
\end{verbatim}
and {\tt Walnut} returns
{\tt TRUE}.

\bigskip

(b) is handled similarly, with the following {\tt Walnut} code.
\begin{verbatim}
def pref1_f "?msd_fib $rep1_f(0,n)":
def pref2_f "?msd_fib $rep2_f(0,n)":
def pref3_f "?msd_fib $rep3_f(0,n)":
eval thm3b "?msd_fib An $pref3_f(n) <=> $pref2_f(n)":
combine FPRF pref2_f=2 pref1_f=1:
\end{verbatim}
The resulting automata is displayed in Figure~\ref{fig3}.   Here we have removed a ``dead state'', state 3, corresponding to illegal Fibonacci representations.
\begin{figure}[htb]
\begin{center}
   \includegraphics[width=6in]{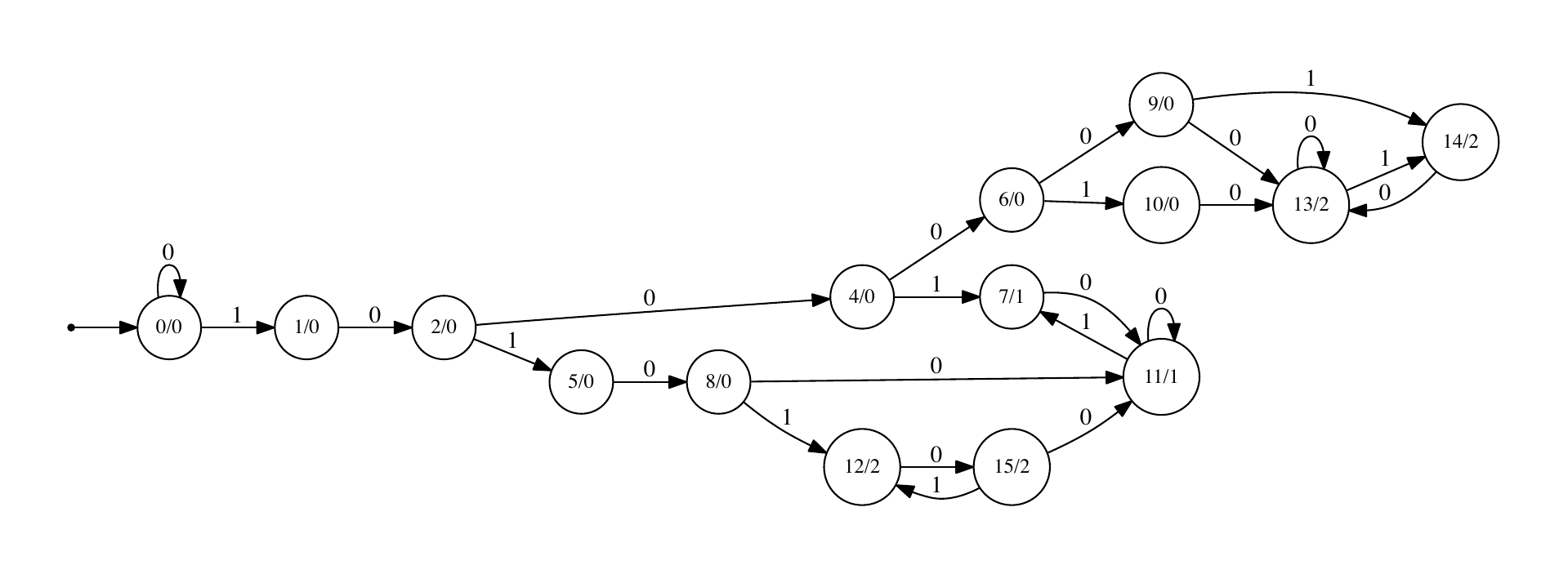}
\end{center}
\caption{DFAO to compute the shortest width of the prefix  ${\bf f}[0..n-1]$.}
\label{fig3}
\end{figure}

\end{proof}

\begin{corollary}
    There is a DFAO that, on input $i$ and $n$ in parallel in Zeckendorf representation, returns the width of the shortest repetition factorization of ${\bf f}[i..i+n-1]$, or $0$ if there is no such
    factorization.
\end{corollary}
\begin{proof}
To create this DFAO, we use the {\tt combine} command of {\tt Walnut},
which uses a familiar direct product construction on automata.
\begin{verbatim}
combine FIBFAC rep3_f=3 rep2_f=2 rep1_f=1:
\end{verbatim}
It creates a $107$-state DFAO, but the figure is too large to present here. 
\end{proof}

\begin{remark}
Inoue et al.~\cite{Inoue:2022} determined the shortest factorization for Fibonacci strings.  This corresponds to the special case of Figure~\ref{fig3}, where $n$ is a Fibonacci number.
\end{remark}

\section{Results about the Thue-Morse word}

The Thue-Morse word $\bf t$ is defined to be the
fixed point of the morphism $0 \rightarrow 01$,
$1 \rightarrow 10$.    In this section, we investigate repetition
factorizations of factors of $\bf t$.   

If we look at such factorizations empirically, we quickly see that
factors can either not have a repetition factorization, or have
one of widths $1,2,\ldots, 7$.   For example, see Table~\ref{tab1}.
\begin{table}[H]
\begin{center}
\begin{tabular}{c|c|c|c}
$i$ & $n$ & $T[i..i+n-1]$ & width \\
\hline
1 & 2 & $1^2$ & 1\\
5 & 4 & $0^2 1^2$ & 2\\
5 & 6 & $0^2 1^2 0^2$ & 3 \\
5 & 12 & $0^2 1^2 0^2 (101)^2$ & 4 \\
5 & 14 & $0^2 1^2 0^2 (101)^2 0^2$ & 5 \\
45 & 20 & $0^2 (101)^2 0^2 1^2 0^2 (101)^2$ & 6 \\
45 & 22 & $0^2 (101)^2 0^2 1^2 0^2 (101)^2 0^2$ & 7 
\end{tabular}
\end{center}
\caption{Factorizations of widths $1$ through $7$.}
\label{tab1}
\end{table}
We now prove that indeed, only widths $1$ through $7$ are possible.

\begin{theorem}
If a factor of the Thue-Morse word has a repetition
factorization, then its width is at most $7$.
\begin{proof}
The idea is the same as for the Fibonacci word.   

\begin{verbatim}
def per_t "(p>=1) & Aj (j+p<n) => T[i+j]=T[i+j+p]":
# 31 states
def rep1_t "Ep $per_t(i,n,p) & 2*p<=n":
# 5 states
def rep2_t "$rep1_t(i,n) | Ej j>=1 & j<n & $rep1_t(i,j) & $rep1_t(i+j,n-j)":
# 8 states
def rep3_t "$rep2_t(i,n) | Ej j>=1 & j<n & $rep1_t(i,j) & $rep2_t(i+j,n-j)":
# 11 states
def rep4_t "$rep3_t(i,n) | Ej j>=1 & j<n & $rep1_t(i,j) & $rep3_t(i+j,n-j)":
# 12 states
def rep5_t "$rep4_t(i,n) | Ej j>=1 & j<n & $rep1_t(i,j) & $rep4_t(i+j,n-j)":
# 14 states
def rep6_t "$rep5_t(i,n) | Ej j>=1 & j<n & $rep1_t(i,j) & $rep5_t(i+j,n-j)":
# 15 states
def rep7_t "$rep6_t(i,n) | Ej j>=1 & j<n & $rep1_t(i,j) & $rep6_t(i+j,n-j)":
# 14 states
def rep8_t "$rep7_t(i,n) | Ej j>=1 & j<n & $rep1_t(i,j) & $rep7_t(i+j,n-j)":
# 14 states
eval thm5 "Ai,n $rep7_t(i,n) <=> $rep8_t(i,n)":
\end{verbatim}
and {\tt Walnut} returns
{\tt TRUE}. 
\end{proof}
\end{theorem}

\begin{corollary}
    If a factor of the Thue-Morse word has a repetition factorization, then this factorization is unique.
\end{corollary}

\begin{proof}
We use the idea in Remark~\ref{rmk1}, and the 
following {\tt Walnut} code.   Here we use the fact that $B = 7$.
\begin{verbatim}
def firstfac_t "t>=1 & t<=n & $rep1_t(i,t) & $rep6_t(i+t,n-t)":
def twofac_t "Et1,t2 t1!=t2 & $firstfac_t(i,n,t1) & $firstfac_t(i,n,t2)":
eval checkuniq_t "~Ei,n $twofac_t(i,n)":
\end{verbatim}
and {\tt Walnut} returns {\tt TRUE}.

\end{proof}

\begin{corollary}
    There is a $28$-state DFAO that, on input $i$ and $n$ in parallel in base $2$, returns the (unique) width of the repetition factorization of ${\bf t}[i..i+n-1]$, or $0$ if there is no such factorization.
\end{corollary}

\begin{proof}
To create this DFAO, we use the {\tt combine} command of {\tt Walnut},
which uses a familiar direct product construction on automata.
\begin{verbatim}
combine TMFAC rep7_t=7 rep6_t=6 rep5_t=5 rep4_t=4 rep3_t=3 rep2_t=2 rep1_t=1:
\end{verbatim}
It creates a $28$-state DFAO, depicted in Figure~\ref{fig2}.  To make
it more readable, the dead state, numbered $5$, with output $0$, has
been omitted.   All omitted transitions go to this state.
\end{proof}
\begin{figure}[H]
\begin{center}
   \includegraphics[width=6in]{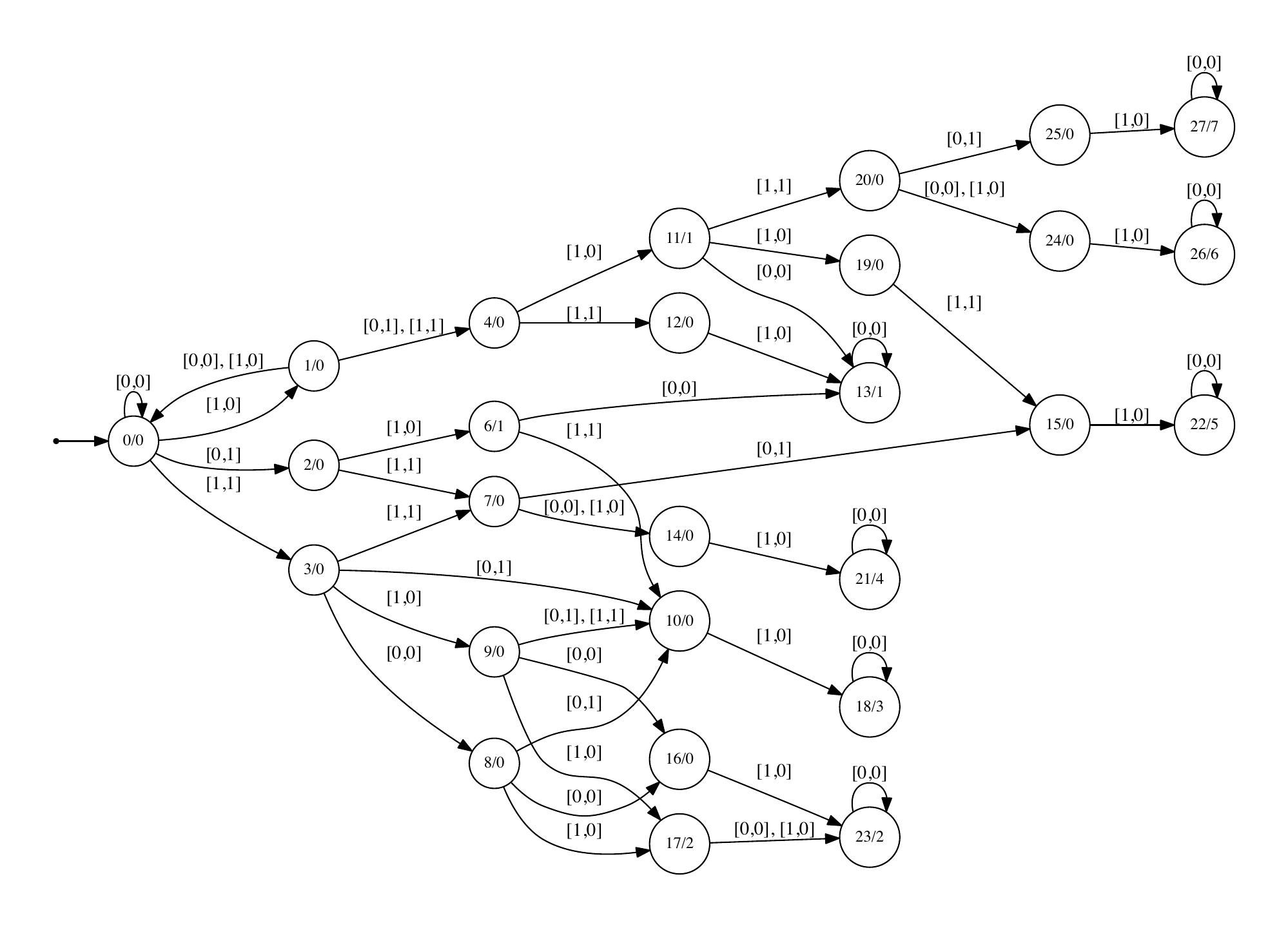}
\end{center}
\caption{DFAO to compute the (unique) width of the factorization of a Thue-Morse factor ${\bf t}[i..i+n-1]$.}
\label{fig2}
\end{figure}

\section{Results about paperfolding words}

The regular paperfolding word can be defined as the image, under the coding $n \rightarrow \lfloor n/2\rfloor$, of the fixed point of the morphism $0 \rightarrow 01$, $1 \rightarrow 21$, $2 \rightarrow 03$, $3 \rightarrow 23$.  For more about the paperfolding sequences, see
\cite[pp.~155--157, 181--185]{Allouche&Shallit:2003}.
In this section we can use our methods to investigate the shortest repetition factorizations of the factors of this word.
\begin{theorem}
    If a factor of the regular paperfolding word has a repetition
    factorization, then its width is at most $10$.
\end{theorem}

\begin{proof}
We use the following {\tt Walnut} code:
\begin{verbatim}
def per_p "(p>=1) & Aj (j+p<n) => P[i+j]=P[i+j+p]":
# 78 states
def rep1_p "Ep $per_p(i,n,p) & 2*p<=n":
# 14 states
def rep2_p "$rep1_p(i,n) | Ej j>=1 & j<n & $rep1_p(i,j) & $rep1_p(i+j,n-j)":
# 37 states
def rep3_p "$rep2_p(i,n) | Ej j>=1 & j<n & $rep1_p(i,j) & $rep2_p(i+j,n-j)":
# 46 states
def rep4_p "$rep3_p(i,n) | Ej j>=1 & j<n & $rep1_p(i,j) & $rep3_p(i+j,n-j)":
# 62 states
def rep5_p "$rep4_p(i,n) | Ej j>=1 & j<n & $rep1_p(i,j) & $rep4_p(i+j,n-j)":
# 70 states
def rep6_p "$rep5_p(i,n) | Ej j>=1 & j<n & $rep1_p(i,j) & $rep5_p(i+j,n-j)":
# 84 states
def rep7_p "$rep6_p(i,n) | Ej j>=1 & j<n & $rep1_p(i,j) & $rep6_p(i+j,n-j)":
# 87 states
def rep8_p "$rep7_p(i,n) | Ej j>=1 & j<n & $rep1_p(i,j) & $rep7_p(i+j,n-j)":
# 93 states
def rep9_p "$rep8_p(i,n) | Ej j>=1 & j<n & $rep1_p(i,j) & $rep8_p(i+j,n-j)":
# 95 states
def rep10_p "$rep9_p(i,n) | Ej j>=1 & j<n & $rep1_p(i,j) & $rep9_p(i+j,n-j)":
# 96 states
def rep11_p "$rep10_p(i,n) | Ej j>=1 & j<n & $rep1_p(i,j) & $rep10_p(i+j,n-j)":
# 96 states
eval thm9 "Ai,n $rep10_p(i,n) <=> $rep11_p(i,n)":
\end{verbatim}
\end{proof}

\begin{corollary}
    There is a DFAO that, on input $i$ and $n$ in parallel in base $2$, returns the  width of the shortest repetition factorization of ${\bf rp}[i..i+n-1]$, or $0$ if there is no such factorization.
\end{corollary}
\begin{proof}
To create this DFAO, we use the {\tt combine} command of {\tt Walnut},
which uses a familiar direct product construction on automata.
\begin{verbatim}
combine RFFAC rep10_p=10 rep9_p=9 rep8_p=8 rep7_p=7 rep6_p=6
    rep5_p=5 rep4_p=4 rep3_p=3 rep2_p=2 rep1_p=1:
\end{verbatim}
It creates a $164$-state DFAO, too large to present here. 
\end{proof}

More generally, one could consider the (uncountable) set of all paperfolding
words.  Each paperfolding word is specified by a sequence of unfolding instructions.   

\begin{theorem}
Every factor of every paperfolding sequence, if it has a repetition factorization, has one of width at most $10$.
\end{theorem}

\begin{proof}
We can use {\tt Walnut} 
to handle logical formulas involving quantification over
all paperfolding words.     We refer
to \cite[Chap.~12]{Shallit:2022} for the full details.

The following {\tt Walnut} code performs the verification.
\begin{verbatim}
reg link {-1,0,1} lsd_2 "([-1,1]|[1,1])*[0,0]*":

def paper "?lsd_2 Ex $link(f,x) & i>=1 & i+n<=x+1 & Aj (j+p<n) =>
   FOLD[f][i+j]=FOLD[f][i+j+p]":
# 166 states

def rep1 "?lsd_2 Ex,p p>=1 & $link(f,x) & i>=1 & i+n<=x+1 & 
   $paper(f,i,n,p) & 2*p<=n":
# 55 states

def rep2 "?lsd_2  Ex $link(f,x) & i>=1 & i+n<=x+1 & ($rep1(f,i,n) |
   Ej j>=1 & j<n & $rep1(f,i,j) & $rep1(f,i+j,n-j))":
# 106 states

def rep3 "?lsd_2  Ex $link(f,x) & i>=1 & i+n<=x+1 & ($rep2(f,i,n) |
   Ej j>=1 & j<n & $rep1(f,i,j) & $rep2(f,i+j,n-j))":
# 214 states

def rep4 "?lsd_2  Ex $link(f,x) & i>=1 &  i+n<=x+1 & ($rep3(f,i,n) |
   Ej j>=1 & j<n & $rep1(f,i,j) & $rep3(f,i+j,n-j))":
# 303 states

def rep5 "?lsd_2  Ex $link(f,x) & i>=1 & i+n<=x+1 & ($rep4(f,i,n) |
   Ej j>=1 & j<n & $rep1(f,i,j) & $rep4(f,i+j,n-j))":
# 367 states

def rep6 "?lsd_2  Ex $link(f,x) & i>=1 & i+n<=x+1 & ($rep5(f,i,n) |
   Ej j>=1 & j<n & $rep1(f,i,j) & $rep5(f,i+j,n-j))":
# 425 states

def rep7 "?lsd_2  Ex $link(f,x) & i>=1 &  i+n<=x+1 & ($rep6(f,i,n) |
   Ej j>=1 & j<n & $rep1(f,i,j) & $rep6(f,i+j,n-j))":
# 467 states

def rep8 "?lsd_2  Ex $link(f,x) & i>=1 &  i+n<=x+1 & ($rep7(f,i,n) |
   Ej j>=1 & j<n & $rep1(f,i,j) & $rep7(f,i+j,n-j))":
# 487 states

def rep9 "?lsd_2  Ex $link(f,x) & i>=1 & i+n<=x+1 & ($rep8(f,i,n) |
   Ej j>=1 & j<n & $rep1(f,i,j) & $rep8(f,i+j,n-j))":
# 485 states

def rep10 "?lsd_2  Ex $link(f,x) & i>=1 & i+n<=x+1 & ($rep9(f,i,n) |
   Ej j>=1 & j<n & $rep1(f,i,j) & $rep9(f,i+j,n-j))":
# 487 states

def rep11 "?lsd_2  Ex $link(f,x) & i>=1 &  i+n<=x+1 & ($rep10(f,i,n) |
   Ej j>=1 & j<n & $rep1(f,i,j) & $rep10(f,i+j,n-j))":
# 487 states

eval thm12 "?lsd_2 Af,i,n $rep10(f,i,n) <=> $rep11(f,i,n)":
\end{verbatim}
and the last command returns {\tt TRUE}.
\end{proof}

One could also ask if every infinite paperfolding sequence contains
a factor of shortest width $10$.   Perhaps surprisingly, this is not
the case; some paperfolding sequences, such as the one
coded by $\overline{1} 111\cdots$, have the property that all factors
having a repetition factorization have one of width at most $7$.  We
can use {\tt Walnut} to prove the following theorems.

\begin{theorem}
Let $u \in \{ -1, +1 \}$ be a finite sequence of unfolding instructions.
and let ${\bf f}_u$ be the finite paperfolding word coded by $u$.
Consider the finite factors $x$ of ${\bf f}_u$ having a repetition factorization.
If $|u| \geq 8$, then either
\begin{itemize}
\item  All finite factors $x$ of ${\bf f}_u$ having a repetition
factorization satisfy $\sw(x) \leq 7$; or

\item There is at least one factor $x$ of ${\bf f}_u$ with $\sw(x) = 10$.
\end{itemize}
Hence the same dichotomy holds for all infinite paperfolding sequences.
\end{theorem}

\begin{proof}
We can prove this with the following {\tt Walnut} code:
\begin{verbatim}
def exact10 "?lsd_2 $rep10(f,i,n) & ~$rep9(f,i,n)":
# 58 states

def thm13 "?lsd_2 Af,x ($link(f,x) & x>=255) => 
   ((Ai,n $rep10(f,i,n) => $rep7(f,i,n))|(Ei,n $exact10(f,i,n)))":
\end{verbatim}
and {\tt Walnut} returns {\tt TRUE}.
\end{proof}

Furthermore, we can create an automaton that accepts exactly those finite
sequences of unfolding instructions, over $\{ -1, +1 \}$, with the 
property that case (a) holds.  To do so, we use the following
{\tt Walnut} code:
\begin{verbatim}
reg good {-1,0,1} "(1|[-1])*":
def seven "?lsd_2 $good(f) & Ai,n $rep10(f,i,n) => $rep7(f,i,n)":
\end{verbatim}
The result is displayed in Figure~\ref{sev}.
\begin{figure}[htb]
\begin{center}
   \includegraphics[width=6.5in]{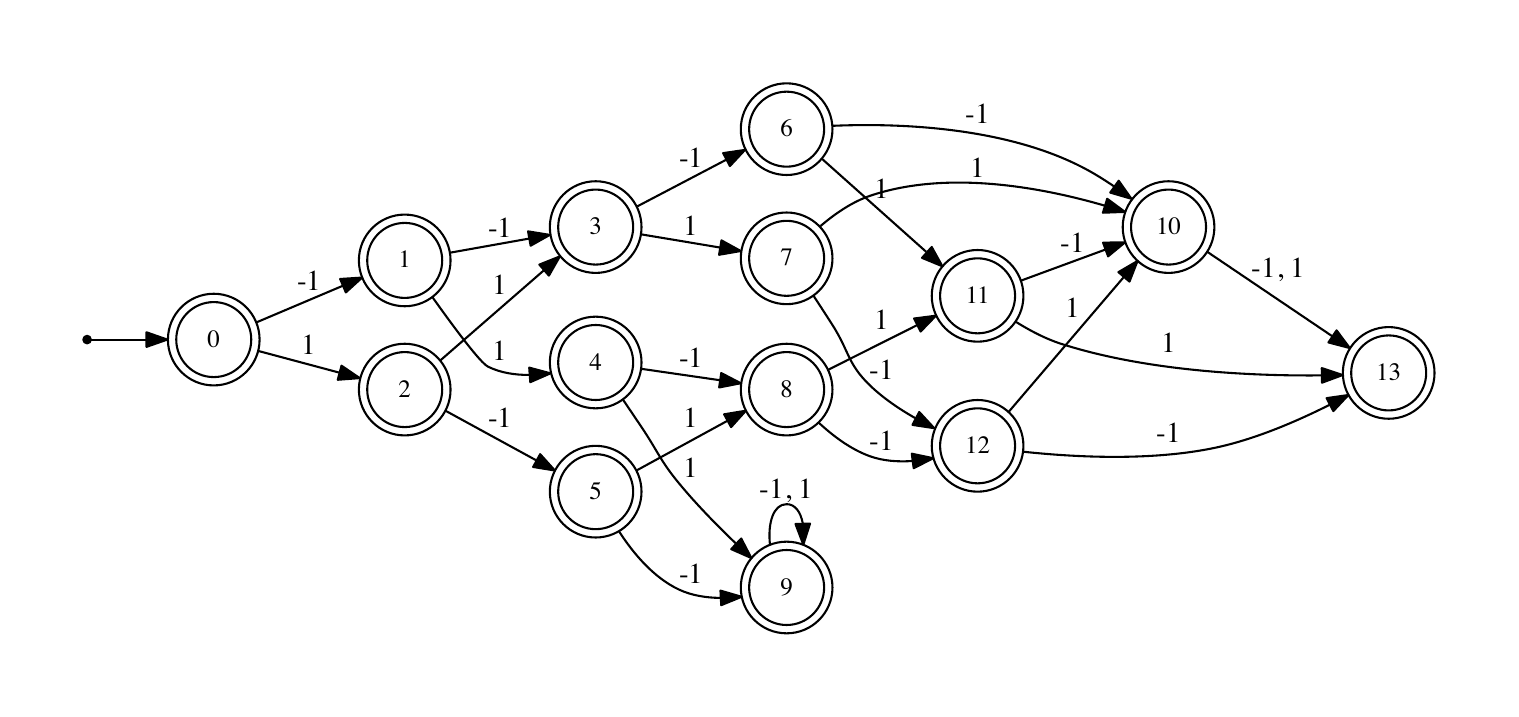}
\end{center}
\caption{Unfolding instructions for paperfolding sequences where all factors $x$
have $\sw(x) \leq 7$.}
\label{sev}
\end{figure}

\begin{proposition}
Every factor of a paperfolding word having a repetition factorization
is of length at most $45$.
\end{proposition}

\begin{proof}
We can ask {\tt Walnut} to compute the lengths of all factors of paperfolding
words having a repetition factorization, has follows:
\begin{verbatim}
def lengths "?lsd_2 Ef,i $rep10(f,i,n)":
\end{verbatim}
The resulting automaton, displayed in Figure~\ref{fig22}, accepts
no number larger than $45$.  (Recall that in this case, the automaton
is accepting the base-$2$ representation of numbers starting with the
{\it least\/} significant digit.)
\begin{figure}[htb]
\begin{center}
   \includegraphics[width=6in]{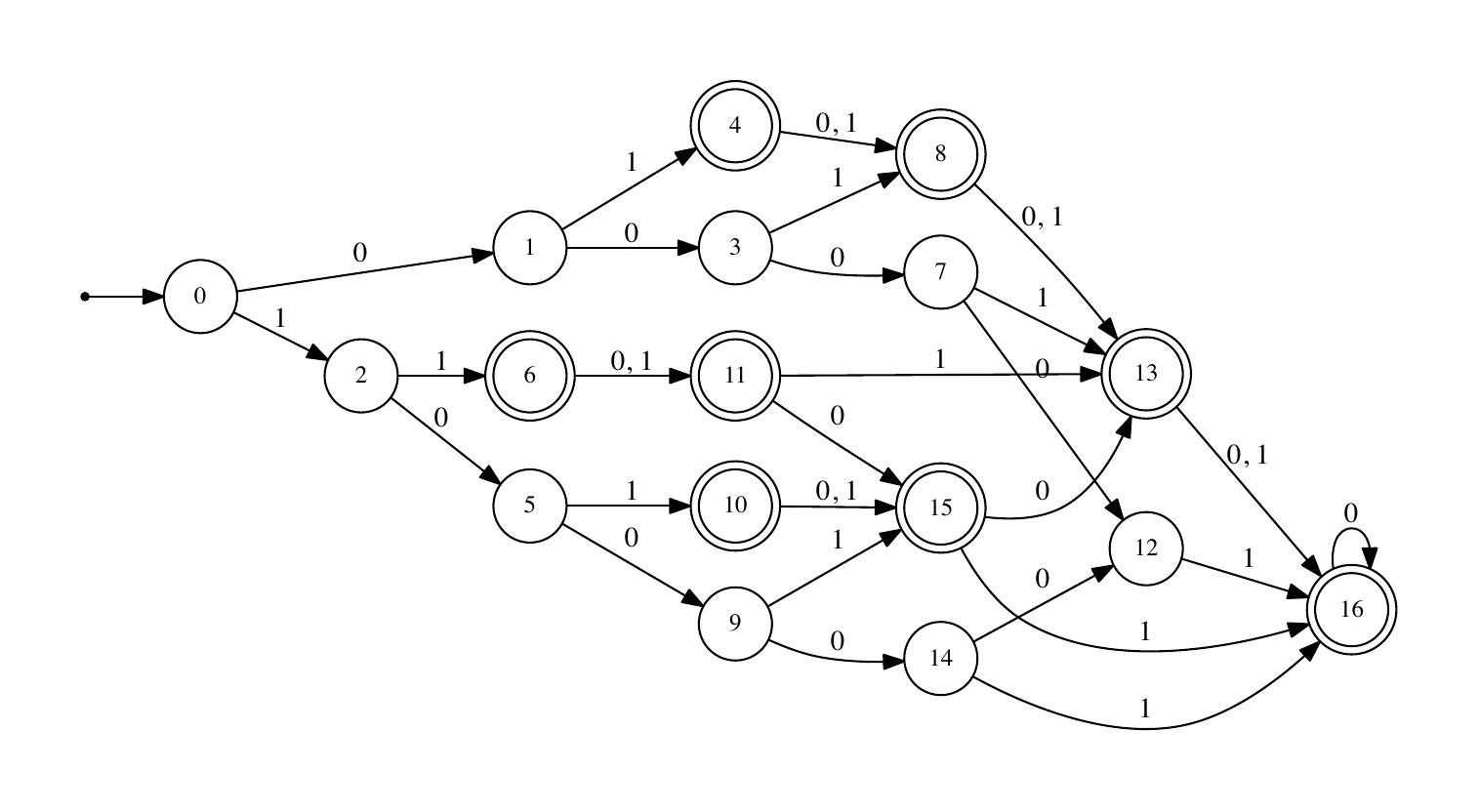}
\end{center}
\caption{Possible lengths of factors of a paperfolding sequence with a repetition factorization.}
\label{fig22}
\end{figure}
\end{proof}

\section{Results about the Rudin-Shapiro sequence}

In some cases we can also use {\tt Walnut} to prove results
about longest factorizations.   This section gives an example.

\begin{theorem}
    There are factors of the Rudin-Shapiro sequence  with repetition factorization of arbitrarily large width.
\end{theorem}

\begin{proof}
    Let's guess that arbitrarily long prefixes of Rudin-Shapiro have repetition factorizations where
each term is of length $\leq L$, for some $L$ to be determined.

Now empirically test this by computing the repetition factorizations for various $L$ up to the first $n$ terms, for $n$ large.

We quickly see that, empirically, it suffices to take $L = 8$.

Now we guess that there is an automaton that, on input $n$ in base $2$, accepts iff ${\bf rs}[0..n-1]$ is factorizable into
repetitions of length $\leq 8$.   ``Guess" this automaton using the Myhill-Nerode theorem.
The guess is an automaton of $23$ states, called {\tt rsrf}.  

Once we have the guess, we can verify its correctness as follows using Walnut:
\begin{verbatim}
def per_r "p>0 & p<=n & Aj (j>=i & j+p<i+n) => RS[j]={RS[j+p]"::
def rep_r "Ep p>=1 & p<=j & $per_r(i,j,p) & 2*p<=j":
eval test "An,i (n>8 & i<=8 & $rsrf(n-i) & $rep_r(n-i,i)) => $rsrf(n)":
\end{verbatim}

Here 
\begin{itemize}
    \item {\tt per\_r}$(i,n,p)$ asserts that ${\tt rs}[i..i+n-1]$ has period $p$
\item {\tt rep\_r}$(i,n)$ asserts that ${\tt rs}[i..i+n-1]$ is a repetition
\item {\tt test} asserts that if ${\tt rs}[0..n-i-1]$ is factorizable into repetitions of length $\leq 8$ and ${\tt rs}[n-i..n-1]$ is a repetition,
then ${\tt rs}[0..n-1]$ is also so factorizable. 
\end{itemize}
This provides the induction step of an induction proof that our
automaton {\tt rsrf} is actually correct.
(The base case also needs to be verified, but this is easy).
Walnut responds {\tt TRUE} for the last assertion, so {\tt rsrf} is correct.

Now we just need to verify that ${\tt rsrf}(n)$ is true for arbitrarily large $n$:
\begin{verbatim}
eval test2 "An Em (m>n) & $rsrf(m)":
\end{verbatim}
which also evaluates to {\tt TRUE}.
\end{proof}

\section{Open problem}

We do not know if it is decidable, given a (generalized) automatic sequence, whether it has words with arbitrarily large width of shortest factorizations.

\end{document}